\documentclass[svgnames,dvipsnames]{llncs}
\usepackage{todonotes}
\usepackage{amsmath,amssymb}
\newtheorem{Reduction Rule}{Reduction Rule}
\newtheorem{claim1}{Claim}
\usepackage{mathtools}
\usepackage{hyperref}

\usepackage{enumerate}

\usepackage{xspace}
\usepackage{graphicx}
\usepackage{grffile}

\usepackage{algorithm2e}
\usepackage{algorithmic}

\newcommand{\F}{\ensuremath{\mathbb F}}

\newcommand{\ICpartition}{IC-partition}

\newcommand{\containment}{\NP~$\subseteq$~\coNP/poly}

\newcommand{\assume}{\P~$\neq$~\NP}
\newcommand{\equality}{\P~$=$~\NP}
\newcommand{\Oh}{\mathcal{O}}

\newcommand{\rlpartization}{{\sc Vertex Partization}}

\newcommand{\rldeletion}{$(r,\ell)$-{vertex deletion set}}
\newcommand{\rlgraph}{$(r,\ell)$-{graph}}
\newcommand{\rlpartition}{$(r,\ell)$-{partition}}
\newcommand{\rlsplitgraph}{$(r,\ell)$-{split graph}}
\newcommand{\rlsplitpartition}{$(r,\ell)$-{\textsc split partition}}

\newcommand{\YES}{\textsc{YES}}
\newcommand{\NO}{\textsc{NO}}

\newcommand{\cochromp}{\textsc{Cochromatic Number on Perfect Graphs}}
\usepackage{amsfonts}
\usepackage{amstext}
\usepackage{graphicx,tikz}
\RequirePackage{fancyhdr}
\usepackage{xcolor}
\usepackage{boxedminipage}
\newcommand{\defparproblem}[4]{
  \vspace{1mm}
\noindent\fbox{
  \begin{minipage}{0.96\textwidth}
  \begin{tabular*}{\textwidth}{@{\extracolsep{\fill}}lr} #1  & {\bf{Parameter:}} #3
\\ \end{tabular*}
  {\bf{Input:}} #2  \\
  {\bf{Question:}} #4
  \end{minipage}
  }
  \vspace{1mm}
}

\newcommand{\defparproblemoutput}[4]{
  \vspace{1mm}
\noindent\fbox{
  \begin{minipage}{0.96\textwidth}
  \begin{tabular*}{\textwidth}{@{\extracolsep{\fill}}lr} #1  & {\bf{Parameter:}} #3
\\ \end{tabular*}
  {\bf{Input:}} #2  \\
  {\bf{Output:}} #4
  \end{minipage}
  }
  \vspace{1mm}
}

\newcommand{\defproblem}[3]{
  \vspace{1mm}
\noindent\fbox{
  \begin{minipage}{0.96\textwidth}
  \begin{tabular*}{\textwidth}{@{\extracolsep{\fill}}lr} #1 \\ \end{tabular*}
  {\bf{Input:}} #2  \\
  {\bf{Question:}} #3
  \end{minipage}
  }
  \vspace{1mm}
}

\newcommand{\NP}{\text{\normalfont  NP}\xspace}
\renewcommand{\P}{\text{\normalfont  P}\xspace}
\newcommand{\coNP}{\text{\normalfont co-NP}\xspace}
\newcommand{\FPT}{\text{\normalfont FPT}\xspace}

\title{Parameterized Algorithms on Perfect Graphs for deletion to $(r,\ell)$-graphs}
\author{Sudeshna Kolay	\inst{1} \and Fahad Panolan\inst{1} \and Venkatesh Raman\inst{1} \and Saket Saurabh \inst{1,2} }

\institute{The Institute of Mathematical Sciences, Chennai, India.  
\email{\{skolay|fahad|vraman|saket\}@imsc.res.in}
\and
{University of Bergen, Norway.  
\email{Saket.Saurabh@uib.no}
}}
\date{}
\begin{document}
\maketitle

\begin{abstract}
For  fixed integers $r,\ell \geq 0$, a graph $G$ is called an {\em $(r,\ell)$-graph} if the vertex set $V(G)$ can be partitioned into $r$ independent sets and $\ell$ cliques. Such a graph is also said to have {\it cochromatic number} $r+\ell$. The class of $(r, \ell)$ graphs generalizes $r$-colourable graphs (when $\ell =0)$ and hence not surprisingly, determining whether a given graph is an $(r, \ell)$-graph is \NP-hard even when $r \geq 3$ or $\ell \geq 3$ in general graphs. 

When $r$ and $\ell$ are part of the input, then the recognition problem is \NP-hard even if the input graph is a perfect graph (where the {\sc Chromatic Number} problem is solvable in polynomial time).
It is also known to be fixed-parameter tractable (FPT) on perfect graphs when parameterized by $r$ and $\ell$. I.e. there is an $f(r+\ell) \cdot n^{\Oh(1)}$ algorithm on perfect graphs on $n$ vertices where $f$ is some (exponential) function of $r$ and $\ell$. Observe that such an algorithm is unlikely on general graphs as the problem is \NP-hard even for constant $r$ and $\ell$.

In this paper, we consider the parameterized complexity of the following problem, which we call {\sc Vertex Partization}. Given a perfect graph $G$ and positive integers $r,\ell,k$  
decide whether there exists a set $S\subseteq V(G)$ of size at most $k$ such that the deletion of $S$ from $G$ results in an $(r,\ell)$-graph. This problem generalizes well studied problems such as {\sc Vertex Cover} (when $r=1$ and $\ell=0$), 
{\sc Odd Cycle Transversal} (when $r=2$, $\ell=0$) and {\sc Split Vertex Deletion} (when $r=1=\ell$).
We obtain the following results:
\begin{enumerate}
\item {\sc Vertex Partization} on perfect graphs is \FPT when parameterized by $k+r+\ell$.
\item The problem does not admit any polynomial sized kernel when parameterized by $k+r+\ell$. 
In other words, in polynomial time, the input graph can not be compressed to an equivalent instance of size polynomial in $k+r+\ell$. In fact, our result holds even when $k=0$.

\item When $r,\ell$ are universal constants, then {\sc Vertex Partization} on perfect graphs, parameterized by $k$, has a polynomial sized kernel.
\end{enumerate}

\end{abstract}

\section{Introduction}

For  fixed integers $r,\ell \geq 0$, a graph $G$ is called an {\em \rlgraph} if the vertex set $V(G)$ can be partitioned into $r$ independent sets and $\ell$ cliques. Many special subclasses of this graph class, such as bipartite, chordal, interval, split and permutation, are well studied in various areas of algorithm design and intractability. For example, $(2,0)$- and  $(1,1)$-graphs correspond to bipartite graphs and split graphs respectively.  
A $(3,0)$-graph is a $3$-colourable graph. Hence, there is a rich dichotomy even with respect to recognition algorithms for \rlgraph{s}. It is well known that we can recognize $(2,0)$- and  $(1,1)$-graphs, on $n$ vertices and $m$ edges, in $\Oh(m+n)$ time. In fact, one can show that recognizing whether a graph $G$ is an \rlgraph, when $r,\ell \leq 2$, can be done in polynomial time \cite{Brandstadt98thecomplexity,Feder03listpartitions}. On the other hand, when either $r \geq 3$ or $\ell \geq 3$, the recognition problem is as hard as the celebrated 
{\sc $3$-Colouring} problem, which is \NP-complete~\cite{Garey:1979:CIG:578533}. Thus, the following problem is \NP-hard:

\defproblem{{\sc Partization Recognition}}{A graph $G$ and positive integers $r,\ell$}{Is $G$ an \rlgraph?} 

{\sc Partization Recognition} has also been studied when the input is restricted to a chordal graph. This restricted problem has a polynomial time algorithm~\cite{DBLP:journals/endm/FederHR11}. On the other hand, when the input graphs are restricted to perfect graphs, {\sc Partization Recognition} is \NP-complete. It was shown in \cite{HeggernesKLRS13}, that the problem, when parameterized by $r+\ell$, has an \FPT{} algorithm, i.e, an algorithm that runs in time $f(r+\ell)\cdot n^{\Oh(1)}$ for a computable function $f$. A natural extension to {\sc Partization Recognition} is the \rlpartization{} problem. The problem is formally stated below:

\defparproblem{{\rlpartization}}{A graph $G$ and positive integers $r,\ell,k$}{$r,\ell,k$}{Is there a vertex subset $S\subseteq V(G)$ of size at most $k$ such that $G-S$ is an \rlgraph?} 
 \smallskip

%

\noindent 
Because of the \NP-hardness of the {\sc $3$-Colouring} problem, we do not expect to obtain an \FPT{} algorithm on general graphs, parameterized by $k+r+\ell$, for \rlpartization, when $r\geq 3$ or $\ell \geq 3$. It has been shown in \cite{KolayP15,DBLP:journals/corr/BasteFKS15} that, for all other combinations of $r$ and $\ell$, namely when $0\leq r,\ell \leq 2$, \rlpartization{} has an \FPT{} algorithm with running time $3.3146^k  n^{\Oh(1)}$. Various special cases of this problem are very well studied. 
When $r=1$ and $\ell=0$, the problem is the same as the celebrated {\sc Vertex Cover} problem, which has been extensively studied in parameterized complexity, and the current fastest algorithm runs in time $1.2738^kn^{\Oh(1)}$ and has a kernel with $2k$ vertices~\cite{DBLP:conf/mfcs/ChenKX06} 
(A brief overview of paramaterized 
complexity is given in the preliminaries).

For $r=2$ and $\ell=0$, the problem is the same as {\sc Odd Cycle Transversal} (OCT) whose parameterized complexity was settled 
by Reed et al.~\cite{ReedSV04} which also used the iterative compression technique 
(which is also the technique we use to show that \rlpartization{} on perfect graphs has an {\FPT} algorithm in this paper) for the first time.
The current best algorithm for the problem is by Lokshtanov et al.~\cite{LokshtanovNRRS14} with a running time of 
$2.3146^kn^{\Oh(1)}$ that uses a branching algorithm based on linear programming.

When $r=\ell=1$, the \rlpartization ~ problem is the same as {\sc Split Vertex Deletion (SVD)} for which 
Ghosh et al.~\cite{DBLP:conf/swat/GhoshKKMPRR12} designed an algorithm with running time
$2^kn^{\Oh(1)}$. 
They also gave the best known polynomial kernel for {\sc SVD}. Later, Cygan and  Pilipczuk~\cite{DBLP:journals/ipl/CyganP13} designed an algorithm for {\sc SVD} running in time $1.2738^{k+o(k)}n^{\Oh(1)}$. 

As in the case of {\sc Partization Recognition}, the {\rlpartization} problem  has  also been studied when 
the input graph is restricted to a non-trivial graph class. By a celebrated result of Lewis and Yannakakis \cite{LewisY80}, this problem is \NP-complete even when restricted to the perfect graph class. In fact, {\sc Odd Cycle Transversal (OCT)} restricted to perfect graphs is \NP-hard, because of this result. Thus, we cannot expect an \FPT{} algorithm for \rlpartization{} parameterized by $r+\ell$, unless \equality. Moreover, because of \NP-hardness of {\sc Partization Recognition} on perfect graphs, we do not expect \rlpartization{} on perfect graphs to be \FPT{} when parameterized by $k$ alone, again under the assumption that \assume. Krithika and Narayanaswamy~\cite{DBLP:journals/jgaa/KrithikaN13} studied {\rlpartization} problems on perfect graphs, and among several results they obtain $(r+1)^k n^{\Oh(1)}$ algorithm for {\sc Vertex $(r,0)$-Partization}  on perfect graphs. In this paper, we generalize this for all values of $r$ and $\ell$. In other words, we show that \rlpartization{} on perfect graphs, parameterized by $k+r+\ell$, is \FPT.

\smallskip

\noindent 
{\bf Our Results and Methods.}  For {\sc Vertex Partization} on perfect graphs, parameterized by $k+r+\ell$, we give an \FPT{} algorithm using the method of iterated compression. This algorithm is inspired by the  
 \FPT{} algorithm for \cochromp, given in \cite{HeggernesKLRS13}. 

Then, we obtain a negative result for {\em kernelization} for \rlpartization{} on perfect graphs. 
We show that \rlpartization{} cannot have a polynomial kernel unless \containment.
This is shown by exhibiting a polynomial parameter transformation from {\sc CNF-SAT}. In fact, 
our result holds even when $k=0$ and either $r$ or $\ell$ is one. Thus, we show that {\sc Partization Recognition} 
(also known as the {\sc Cochromatic Number} problem~\cite{HeggernesKLRS13}) does not admit a polynomial kernel on perfect graphs unless \containment.  See Section~\ref{prelims} 
for the definition of polynomial parameter transformation and kernelization. 

Finally, we consider the following parameterized problem:

\defparproblem{{\sc Vertex $(r,\ell)$ Partization}}{A graph $G$ and a positive integer $k$}{$k$}{Is there a vertex subset $S\subseteq V(G)$ of size at most $k$ such that $G-S$ is an \rlgraph?} 

\noindent 
For each pair of constants $r$ and $\ell$, we give a kernelization algorithm for the above parameterized problem. 
To arrive at the kernelization algorithm, we consider a slightly larger class of graphs called {\em $(r,\ell)$-split graphs}. A graph $G$ is an \rlsplitgraph{} if its vertex set can be partitioned into $V_1$ and $V_2$ such that the size of a largest clique in $G[V_1]$ is bounded by $r$ and the size of a largest independent set in $G[V_2]$ is bounded by $\ell$. Such a two-partition for the graph $G$ is called as \rlsplitpartition. The notion of \rlsplitgraph{s} was introduced in~\cite{Gyarfas98}. 
For any fixed $r$ and $\ell$, there is a finite forbidden set $\F_{r,\ell}$ for \rlsplitgraph{s}~\cite{Gyarfas98}. That is,  
a graph $G$ is an \rlsplitgraph{} if and only if $G$ does not contain any graph $H\in \F_{r,\ell}$ as an induced subgraph. The size of the largest forbidden graph is bounded by $f(r,\ell)$, for some function $f$ only of $r$ and $\ell$~\cite{Gyarfas98}. We use this to design the kernelization algorithm. 

%

\section{Preliminaries}
\label{prelims}
We use $[n]$ to denote $\{1,\ldots,n\}$.
We use standard notations from graph theory \cite{diestel}.  The vertex set and edge set of a graph are  
denoted as $V(G)$ and $E(G)$ respectively. The complement of the graph $G$, denoted by $\overline{G}$ has $V(G)$ as 
its vertex set and ${V(G)\choose 2}-E(G)$ as its edge set.  Here, ${V(G)\choose 2}$ denotes the family of two sized 
subsets of $V(G)$. 
The neighbourhood of a vertex $v$ is represented as $N_G(v)$, or, when the context of the graph is clear, 
simply as $N(v)$. An induced subgraph of $G$ on the vertex set $V'\subseteq V$ is written as $G[V']$. 
For a vertex subset $V' \subseteq V$, $G[V- V']$ is also denoted as $G - V'$. 
We denote by $\omega(G)$ the size of a maximum clique in $G$. Similarly, $\alpha(G)$ denotes the size of a maximum independent set in $G$. The chromatic number of $G$ is denoted by $\chi(G)$.   In this paper, we consider the class of \rlgraph{s}, The following is the formal definition of this graph class.
\begin{definition}[{\rlgraph}]
 A graph $G$ is an \rlgraph{} if its vertex set can be partitioned into $r$ independent sets and $\ell$ 
cliques. We call such a partition of $V(G)$ an \rlpartition. 
\end{definition}
For a graph $G$, we say $S\subseteq V(G)$ is an {\em \rldeletion}, if $G- S$ is an \rlgraph. 

\begin{definition}[{\ICpartition}]
 An \ICpartition, of an \rlgraph{} $G$, is a partition $(V_1,V_2)$ of 
$V(G)$ such that $G[V_1]$ can be partitioned into $r$ independent sets and $G[V_2]$ can be partitioned into $\ell$ cliques.   
\end{definition}

%
A graph $G$ is a {\em perfect graph} if, for every induced subgraph $H$, $\chi (H) = \omega (H)$. We also need the following characterisation of perfect graphs -- also known as strong perfect graph theorem. 
\begin{proposition}[\cite{CRST06}] \label{prop:perfect_graph_char}
A graph $G$ is perfect if and only if $G$ does not have, as an induced subgraph, an odd cycle of length at least $5$ or its complement.  
\end{proposition}
This  tells us that perfect graphs are closed under complementation. 
However, this was proved earlier and was called  weak perfect graph theorem. 
 \begin{lemma}[\cite{Lovasz72}]\label{pg:complement}
 $G$ is a perfect graph if and only if $\overline{G}$ is a perfect graph.
 \end{lemma}

Moreover, this class is well known for its tractability for \NP-hard problems such as   
{\sc Chromatic Number}, {\sc Maximum Independent Set} and {\sc Maximum Clique}.

\begin{lemma}[\cite{GrotschelLS84}]\label{pg:colindclicomp}
Let $G$ be a perfect graph. Then there exist  polynomial time algorithms for {\sc Chromatic Number}, 
{\sc Maximum Independent Set} and {\sc Maximum Clique}. 
\end{lemma}

\begin{proposition} {\rm \bf (\cite[Lemma 3]{HeggernesKLRS13}).}
\label{prop:polytime_algorithms}
Given a perfect graph $G$ and an integer $\ell$, there is a polynomial time algorithm to output
\begin{enumerate}
\vspace{-2mm}
\item[(a)] either a partition of $V(G)$ into at most $\ell$ independent sets or a clique of size $\ell+1$, and
\item[(b)] either a partition of $V(G)$ into at most $\ell$ cliques or an independent set of size $\ell+1$.
\end{enumerate}
\end{proposition}


%
%
%
%
%
%
%
%

%

\noindent
\textbf{Parameterized Complexity.} A parameterized problem $\Pi$ is a subset of $\Gamma^{*}\times\mathbb{N}$, where $\Gamma$ is a finite alphabet. An instance of a parameterized problem is a tuple $(x,k)$, where $x$ is a classical problem instance, and $k$ is called the parameter. A central notion in parameterized complexity is {\em Fixed Parameter Tractability (FPT)}. A parameterized problem $\Pi$ is in \FPT{} if there is an algorithm that takes an instance $(x,k)$ and decides if $(x,k)\in \Pi$ in time $f(k)\cdot |x|^{\Oh(1)}$. Here, $f$ is an arbitrary function of $k$. Such an algorithm is called a \emph{Fixed Parameter Tractable} algorithm and, in short, an \FPT{} algorithm.

\medskip
\noindent
\textbf{Kernelization.} A kernelization algorithm for a parameterized problem 
$\Pi\subseteq \Gamma^{*}\times \mathbb{N}$ is an algorithm that, given $(x,k)\in \Gamma^{*}\times \mathbb{N} $, 
outputs, in time polynomial in $|x|+k$, a pair $(x',k')\in \Gamma^{*}\times \mathbb{N}$ such that 
(a) $(x,k)\in \Pi$ if and only if  $(x',k')\in \Pi$ and (b) $|x'|,k'\leq g(k)$, where $g$ is some computable function. 
The output instance $x'$ is called the kernel, and the function $g$ is referred to as the size of the kernel. 
If $g(k)=k^{\Oh(1)}$  then we say that $\Pi$ has a polynomial  kernel.

\noindent
{\bf Lower bounds in Kernelization.}
In recent years, several techniques have been developed to show that certain parameterized problems  cannot have any polynomial sized kernel unless some classical complexity assumptions are violated. One such technique is 
the {\em polynomial parameter transformation}. 

\begin{definition}
 Let $\Pi,\Gamma$ be two parameterized problems. A polynomial time algorithm $\mathcal A$ is called a polynomial parameter transformation (or ppt) from $\Pi$ to $\Gamma$ if , given an instance $(x,k)$ of $\Pi$, $\mathcal A$ outputs in polynomial time an instance $(x',k')$ of $\Gamma$ such that $(x,k) \in \Pi$ if and only if $(x',k')\in \Gamma$ and $k' \leq k^{\Oh(1)}$.
\end{definition}

We use the following theorem together with ppt reductions to rule out polynomial kernels. 
\begin{proposition}[\cite{BodlaenderTY11}] \label{prop:ppt_redn}
Let $\Pi,\Gamma$ be two parameterized problems such that $\Pi$ is \NP-hard, $\Gamma \in \NP$ and   there exists a polynomial parameter transformation from $\Pi$ to $\Gamma$. Then, if $\Pi$ does not admit a polynomial kernel neither does $\Gamma$.
\end{proposition}

\begin{proposition}[\cite{FortnowS11}] \label{prop:CNF_no_poly_kernel}
 {\sc CNF-SAT} is \FPT{} parameterized by  the number of variables; however, it 
 does not admit a polynomial kernel unless \containment. 
\end{proposition}



\section{\FPT algorithm for {\sc Vertex Partization}}

In this section, we show that \rlpartization{} on perfect graphs is \FPT, using the iterative compression technique. 
Let $(G,r,\ell,k)$ be an input instance of \rlpartization{} on perfect graphs, and let $V(G) =\{v_1,\ldots,v_n\}$. We define, for every $1\leq i \leq n$, the vertex set
$V_i=\{v_1,\ldots,v_i\}$. 
Let $G_i$ denote  $G[V_i]$. Let $i_0= r+\ell+k+1$. We iterate through the instances $(G_i,r,\ell,k)$ starting from $i = i_0$. Given the $i^{th}$ instance and a known \rldeletion{} $S'_i$ of size at most
$k+1$, our objective is to obtain an \rldeletion{} $S_i$ of size at most $k$. The formal definition of this compression problem is as follows.

 \defparproblemoutput{{\sc \rlpartization{} Compression}}
 {A perfect graph $G$, non-negative integers $r,\ell,k$ and a $k+1$-sized vertex subset $S'$ such that $G-S'$ is an 
 \rlgraph, along with an \ICpartition{} $(Q_1,Q_2)$ of $G-S'$.}{$r,\ell,k$}{A vertex subset $S \subseteq V(G)$  
 of size at most $k$ such that $G- S$ is an \rlgraph, and an \ICpartition{} $(P_1,P_2)$ of $G-S$.} \vspace{10 pt}
 
Before we solve {\sc \rlpartization{} Compression}, we explain how to reduce the \rlpartization{} problem 
to $n-(r+\ell+k+1)+1$ instances of the {\sc \rlpartization{} Compression} problem on $G_i$, from $i =i_0$ to $i=n$. 
For the graph $G_{i_0}$, the set $V_{k+1}$ is a \rldeletion{}, $S_i'$, of size $k+1$. The graph $G_{i_0}-V_{k+1}$ has 
$r+\ell$ vertices. We set $Q_1^{i_0}$ to be a set of any $r$ vertices of $V_i-V_{k+1}$ and $Q_2^{i_0}$ to be 
the remaining set of $\ell$ vertices; that is,  $Q_2^{i_0}=V_{i_0}-V_{k+1}-Q_1^{i_0}$. Now, let $I_i = (G_i,r,\ell,k,S'_i,(Q_1^i,Q_2^i))$ be 
the $i^{th}$ instance of {\sc \rlpartization{} Compression}. 
If $S_{i-1}$ is a $k$-sized solution for $I_{i-1}$, then $S_{i-1} \cup \{v_i\}$ is a $(k+1)$-sized \rldeletion{} 
for $G_i$. 
So, the iteration begins with the instance $I_{i_0}= (G_{i_0},r,\ell,k, V_{k+1},(Q_1^{i_0},Q_2^{i_0}))$ and we try 
to obtain an \rldeletion{} of size at most $k$. If such a solution $S_{i_0}$ exists, we set $S'_{i_0 +1} = S_{i_0} \cup
\{v_{i_0+1}\}$ and ask for a $k$-sized solution for the instance $I_{i_0+1}$. We continue in this manner. If, during 
any iteration, the corresponding instance does not have an \rldeletion{} of size at most $k$, then this implies that 
the original instance $(G,r,\ell,k)$ is a \NO{} instance for \rlpartization. 
If $S_n$ is a $k$-sized \rldeletion{} for $G_n$, where $G_n=G$, then clearly 
$(G,r,\ell,k)$ is \YES{} instance of \rlpartization. 
Since there are at most $n-(r+\ell+k+1)+1$ iterations, the total time taken by the algorithm to solve 
\rlpartization{} is at most $n-(r+\ell+k+1)+1$ times the time taken to solve {\sc \rlpartization{} Compression}. 
Thus, if {\sc \rlpartization{} Compression} is \FPT, it follows that \rlpartization{} is \FPT.

We can also {\em view} the input graph $G$ of an instance of {\sc \rlpartization{} Compression} as an $(r+k+1,\ell)$-graph 
with \ICpartition{} $(Q_1\cup S',Q_2)$. Equivalently, {\sc \rlpartization{} Compression} has an input 
positive integers 
$r,\ell,k$ and a graph $G$ that is an $(r+k+1,\ell)$-graph and the objective is to decides whether 
there is a $k$-sized set $S \subseteq V(G)$ such that $G-S$ is an \rlgraph. 
This view point allows us to design an \FPT{} algorithm for {\sc \rlpartization{} Compression}. 
%
%
%
Towards that we first define some notations.
Let $G$ be a graph and $V(G)=\{v_1,\ldots,v_n\}$. 
For partition $P=(A,B)$ of $V(G)$ we define an $n$-length bit vector $B_P^G$ corresponding to 
$P$ as follows. We set the $i^{th}$ bit 0 if $v_i\in A$ and $1$ otherwise. 
For two $n$-length bit vectors $a=a_1\ldots a_n$ and $b=b_1\ldots b_n$, the hamming distance between $a$ and $b$, 
denoted by ${\mathcal H}(a,b)$, is the number of indices on which $a$ and $b$ mismatches. 
That is ${\mathcal H}(a,b)=\vert\{(a_i,b_i)~\vert a_i\neq b_i, i\in [n]~\} \vert$. 
%
%
%
The Hamming distance for the bitvectors corresponding to two \ICpartition{s}, of a graph, 
is bounded as given by the following proposition.

\begin{proposition}[\cite{HeggernesKLRS13}] \label{prop:hamming_distance_bound}
 Let $G$ be a graph.
Let $Q$ be an \ICpartition{} of $G$ that realizes that $G$ is an $(r',\ell')$-graph and  
$P$ is an \ICpartition{} of $G$ that realizes that $G$ is an $(r,\ell)$-graph.  
 Then $\mathcal{H}(B_Q^G,B_P^G) \leq r'\ell +r\ell'$.
\end{proposition}

The following lemma follows from Proposition~\ref{prop:hamming_distance_bound}. 

\begin{lemma}\label{lem:essential_soln}
Let $G$ be a perfect graph and $(Q_1,Q_2)$ be an \ICpartition{} 
that realizes that $G$ is an $(r',\ell')$-graph. 
Let $S$ be a \rldeletion{} for $G$ such that $P$ is \ICpartition{} of $G$ 
that realizes that $G-S$ is an \rlgraph\ and let $Q=(Q_1\setminus S,Q_2\setminus S)$. 
%
Then $\mathcal{H}(B^{G-S}_Q,B^{G-S}_P) \leq r'\ell + r\ell'$
\end{lemma}

Lemma~\ref{lem:essential_soln} implies that to solve {\sc \rlpartization{} Compression} 
it is enough to solve the following problem.  

 \defparproblemoutput{{\sc Short \rlpartization}}{A perfect graph $G$, 
positive integers $r,\ell,k, \rho$, and 
a partition $Q=(Q_1,Q_2)$ of $V(G)$.}{$r,\ell,k,\rho$}{A vertex subset $S \subseteq V(G)$  of size at most $k$ such 
that $G- S$ is a \rlgraph, and an \ICpartition{} $(P_1,P_2)$ of $G-S$ such that 
${\mathcal H}(B_P^{G-S}, B_{Q'}^{G-S})\leq \rho$ where $Q'=(Q_1\setminus S, Q_2\setminus S)$.} \vspace{10 pt}



\begin{lemma}
{\sc Short \rlpartization} is \FPT. 
\end{lemma}

\begin{proof}
We design a recursive algorithm ${\cal A}$ for {\sc Short \rlpartization} which takes a tuple 
$(G,r,\ell,k,\rho, Q=(Q_1,Q_2))$ as input, 
where $G$ is a graph, $Q$ is a partition of $V(G)$ and $r,\ell,k,\rho$ are integers. It outputs a $k$-sized 
\rldeletion{} $S$ of $G$ and an \ICpartition{} $P$ that realizes that $G-S$ is \rlgraph{} such that 
${\mathcal H}(B_P^{G-S},B_{Q'}^{G-S})\leq \rho$, where $Q'=(Q_1\setminus S,Q_2\setminus S)$, 
if such a tuple $(S,P)$ exists, otherwise returns \NO. 
The following are the steps of the recursive algorithm ${\cal A}$ on input 
$(G,r,\ell,k,Q=(Q_1,Q_2),\rho)$. 
\begin{enumerate}
\item[1.] If $k < 0$ or $\rho < 0$ then output \NO.  
\item[2.] If $G[Q_1]$ is $r$-colorable and $\overline{G}[Q_2]$ is $\ell$-colorable, 
then return $(\emptyset, Q)$.
\item[3.] If $G[Q_1]$ is not $r$-colorable, then there is an $r+1$-sized clique in $G_1$. 
By Proposition~\ref{prop:polytime_algorithms}, we can find such a $r+1$-sized clique $C$ 
in polynomial time. Make the following recursive calls to ${\cal A}$:
\begin{itemize}
 \item[(a)] For every vertex $v \in V(C)$, 
do a recursive call ${\cal A}(G-v,r,\ell,k-1 ,\rho,(Q_1\setminus \{v\}, Q_2\setminus \{v\}))$ and if the recursive call returns 
$(S',P)$ then return $(S'\cup \{v\},P)$ as output. 
 \item[(b)] For every vertex $v \in C$, 
do a recursive call ${\cal A}(G,r,\ell,k,\rho-1,(Q_1\setminus \{v\}, Q_2\cup \{v\}))$ and if the recursive call 
returns $(S',P)$ then return $(S',P)$ as output.
\end{itemize}
If all the recursive calls above (in Step $3$) return \NO, then return \NO.
\item[4] If $\overline{G}[Q_2]$ is not $\ell$-colorable, then there is clique of 
size $\ell+1$ in $\overline{G}[Q_2]$. Thus, using Proposition~\ref{prop:polytime_algorithms}, 
we can find an $\ell+1$-sized independent set $I$ in ${G}[Q_2]$. Make the following recursive calls of the algorithm:
\begin{itemize}
\item[(a)] For every vertex $v \in I$, 
do a recursive call ${\cal A}(G-v,r,\ell,k-1,\rho,(Q_1\setminus \{v\}, Q_2\setminus \{v\}))$ and if the recursive 
call returns 
$(S',P)$ then return $(S'\cup \{v\},P)$ as output. 
 \item[(b)] For every vertex $v \in I$, 
do a recursive call ${\cal A}(G,r,\ell,k,\rho-1,(Q_1\cup\{v\}, Q_2\setminus \{v\}))$ and if the recursive call 
returns $(S',P)$ then return $(S',P)$ as output.
\end{itemize}
 If all the recursive calls above (in Step $4$) return \NO, then return \NO.
\end{enumerate}

\noindent
\textbf{Correctness.} 
We now prove that 
the recursive algorithm 
${\cal A}$ is correct. 
We show that if $(G,r,\ell,k,\rho,(Q_1,Q_2))$ is an \YES{} 
instance of  {\sc Short \rlpartization}, 
then the algorithm ${\cal A}$ will output correct solution. We prove this  
by induction on $k+\rho$.
 In the reverse directed we need to show that  if the algorithm ${\cal A}$ returns \NO\ then indeed the given instance is a \NO\ instance. The proof of reverse direction is similar to the proof of forward direction and thus omitted.  

\noindent
\textbf{Base case: $k=0$ and $\rho=0$.}, 
Since $(G,r,\ell,k,\rho,(Q_1,Q_2))$ is an \YES{} instance 
$(Q_1,Q_2)$ is an \ICpartition{} which realizes that 
$G$ is an \rlgraph. In Step 2 of the algorithm ${\cal A}$
we output $(\emptyset,(Q_1,Q_2))$ as the output.

\noindent
\textbf{Induction.} By induction hypothesis we assume 
that ${\cal A}$ outputs correct answer when $k+\rho<\gamma$, where 
$\gamma \geq 0$. 
Now we need to show that ${\cal A}$ outputs correct answer when $k+\rho= \gamma$.   
Let $(G,r,\ell,k,\rho,Q=(Q_1,Q_2))$ be the input of ${\cal A}$ such that $k+\rho=\gamma$. 
Let $(S,(P_1,P_2)$ be a solution of {\sc Short \rlpartization}. 
If $S=\emptyset$ and $(P_1,P_2)=(Q_1,Q_2)$ then in Step 
2, algorithm ${\cal A}$ will output $(\emptyset,Q)$. 
Otherwise either $G[Q_1]$ is not $r$-colorable or $\overline{G}[Q_2]$ is not $\ell$-colorable.

\smallskip
\noindent 
\textbf{Case 1:}
Suppose $G[Q_1]$ is not $r$-colorable. 
Then there is a clique $C$ of size $r+1$ in $G[Q_1]$. 
In this case at least one vertex $v$ from $C$ either belongs to $S$ 
or not belongs to $P_1$. If $v\in S$, then consider the recursive call 
${\cal A}(G-v,r,\ell,k-1,\rho, (Q_1\setminus \{v\},Q_2\setminus \{v\})$. 
By induction hypothesis ${\cal A}(G-v,r,\ell,k-1,\rho, (Q_1\setminus \{v\},Q_2\setminus \{v\})$ 
will return $(S',P')$ such that $S'$ is a $k-1$ sized \rldeletion{} of $G-\{v\}$ such that 
${\mathcal H}(B_{P'}^{G-S'},B_{Q'}^{G-S'})\leq \rho$, where 
$Q'=(Q_1\setminus (S'\cup \{v\}), Q_2\setminus (S'\cup \{v\}))$. 
Hence in Step $3(a)$, algorithm ${\cal A}$ will output $(S'\cup \{v\},P)$  and 
this is a solution for {\sc Short \rlpartization} on input 
$(G,r,\ell,k,\rho,(Q_1,Q_2))$. 

If $v\notin S$, then $v\notin P_1$ as well. Now 
consider the recursive call 
${\cal A}(G,r,\ell,k,\rho-1, (Q_1\setminus \{v\},Q_2\cup \{v\})$. 
By induction hypothesis ${\cal A}(G,r,\ell,k,\rho-1, (Q_1\setminus \{v\},Q_2\cup \{v\})$ 
will return $(S'',P'')$ such that $S''$ is a $k$ sized \rldeletion{} of $G$ such that 
${\mathcal H}(B_{P''}^{G-S''},B_{Q'}^{G-S''})\leq \rho-1$, where 
$Q'=((Q_1\setminus \{v\})\setminus S'', (Q_2\cup \{v\})\setminus S'')$. 
Hence in Step $3(a)$, algorithm ${\cal A}$ will output $(S'',P'')$. 
Clearly $S''$ is a $k$ sized \rldeletion{} of $G''$. Now we show that 
${\mathcal H}(B_{P''}^{G-S''},B_{Q''}^{G-S''})\leq \rho$, where 
$Q''=(Q_1\setminus S'', Q_2\setminus S'')$. Note that 
${\mathcal H}(B_{Q'}^{G-S''}, B_{Q''}^{G-S''})\leq 1$. 
Since ${\mathcal H}(B_{P''}^{G-S''},B_{Q'}^{G-S''})\leq \rho-1$ and 
${\mathcal H}(B_{Q'}^{G-S''}, B_{Q''}^{G-S''})\leq 1$, we have that \\
${\mathcal H}(B_{P''}^{G-S''},B_{Q''}^{G-S''})\leq \rho$

\smallskip
\noindent 
{\bf Case 2 : $\overline{G}[Q_2]$ is not $\ell$-colorable. } 
This case is symmetric to Case 1.

\noindent
\textbf{Running Time.} 
Note that when $k<0$ or $\rho<0$, then the algorithm will stop in a single step. 
Each recursive call either decreases $k$ by $1$ or $\rho$ by $1$. Hence the 
depth of the recursion tree is bounded by $k+\rho+1$. 
Note that in Step $3$ we makes at most $2(r+1)$ recursive calls and in Step 4 
we make at most $2(\ell+1)$ recursive calls. Hence the total 
running time of the algorithm ${\cal A}$ is bounded by $\Oh(\max\{(2(r+1))^{k+\rho+1}n ^{\Oh(1)},(2(\ell+1))^{k+\rho+1}n^{\Oh(1)}\})$. 
\qed
\end{proof}

{\sc \rlpartization{} Compression} is a special case of {\sc Short \rlpartization} 
when $\rho=(r+k+1)\ell+r\ell$. Therefore, we have the following theorem.

\begin{theorem}
\rlpartization{} on perfect graphs is  \FPT.
\end{theorem}

\section{Kernel lower bound}

In this section we show that \rlpartization{} on perfect graphs does not have polynomial kernels. In fact, we show that {\sc Partization Recognition} on perfect graphs cannot have a polynomial kernel, when parameterized by $r+\ell$, unless \containment. 
%
%
%
%
{\sc Partization Recognition} 
on perfect graphs 
when parameterized by $r+\ell$ was shown to be \FPT{} in \cite{HeggernesKLRS13}. Below, we show that the problem cannot have any polynomial  sized kernel. 

\begin{theorem}\label{lem:no_kernel}
{\sc Partization Recognition} 
on perfect graphs 
when parameterized by $r+\ell$, 
does not have a polynomial kernel unless \containment.
\end{theorem}

\begin{proof}
We prove the lemma by  giving a polynomial parameter transformation 
 from {\sc CNF-SAT} parameterized by the number of variables. From Proposition~\ref{prop:CNF_no_poly_kernel}, we know that {\sc CNF-SAT} parametrized by the number of variables does not have a polynomial kernel unless \containment~\cite{FortnowS11}. Then the proof of the theorem follows 
from Proposition~\ref{prop:ppt_redn}. 
We start with an instance $(\phi,n)$ of {\sc CNF-SAT}, where $\phi$ is a CNF formula with $m$ clauses and $n$ variables. Without loss of generality we assume that there is no clause where both literals of a variable appear  together, since such a clause will be satisfied by any assignment and hence can be removed. The polynomial parameter transformation produces an instance $(G,n,1)$ of 
{\sc Partization Recognition}, where $G$ is a perfect graph, such that 
$(\phi,n)$ is an \YES{} instance of {\sc CNF-SAT} if and only if 
$(G,n,1)$ is an \YES{} instance of {\sc Partization Recognition}.  
Let ${\cal C}=\{C_1,\ldots,C_m\}$, $X=\{x_1,\ldots,x_n\}$  and $L=\{x_1,\bar{x_1},\ldots,x_n,\bar{x_n}\}$
be the set of clauses, variables and literals  of $\phi$ respectively. 
The construction the graph $G$ from the formula $\phi$ is as follows (illustrated in Figure~\ref{fig:no_kernel}):
\begin{enumerate}
 \item For each variable $x$, we create two vertices $v_x,v_{\bar{x}}$ which represent the literals $x,\bar{x}$. We call them the literal vertices. More specifically, we call $v_x$ the positive literal vertex and $v_{\bar{x}}$ the negative literal vertex.
 \item For each clause $C$, we create two vertices $w_C^1,w_C^2$. We call these the clause vertices corresponding to the clause $C$. 
 \item For each pair of variables $x,y$, we add the edges  $(v_x,v_y),(v_{\bar{x}},v_y),(v_{\bar{x}},v_y),$ and $(v_{\bar{x}},v_{\bar{y}})$. Notice that $(v_x,v_{\bar{x}})$ and $(v_y,v_{\bar{y}})$ are non-edges. 
 \item For each clause $C$ and each literal $q\in L$, if $q\notin C$,  we add edges 
$(x_q,w_C^1)$ and $(x_q,w_C^2)$. 
In other words if a literal $q'$ belongs to a clause $C$, then 
$(q',w_C^1),(q',w_C^2)\notin E(G)$ In other words there is a complete bipartite graph between $L-C$ and $\{w_{C}^1, w_{C}^2\}$. 
\end{enumerate}
In short, the vertex set and edge set of $G$ is defined as follows 
(note that for a literal $x$, if $x=\bar{y}$, then $y={\bar{x}}$).  
\begin{eqnarray*}
V(G)&=&\{v_x,v_{\bar{x}}~|~x\in X\}\cup \{w_C^1,w_C^2~|~C\in{\cal C}\}\\
E(G)&=&\{(v_{x},v_{y})~|~x\neq \bar{y}\}\cup \{(w_C^1,x),(w_C^2,v_x)~|x\notin C\}
\end{eqnarray*}
\begin{figure}[t] 
   \centering
 \begin{tikzpicture}[scale=0.8]
\node [draw,circle] (a) at (0,0) {};
\node at (0,0.35) {$v_{x_1}$};
\node [draw,circle] (b) at (1,0) {};
\node at (1,0.35) {$v_{\bar{x}_1}$};
\node [draw,circle] (c) at (3,0) {};
\node at (3, 0.35) {$v_{x_2}$};
\node [draw,circle] (d) at (4,0)  {};
\node at (4, 0.35) {$v_{\bar{x}_2}$};
\node [draw,circle] (e) at (0,-2) {};
\node at (0,-2.5) {$w_{C^1_1}$};
\node [draw,circle] (f) at (1,-2) {};
\node at (1,-2.5) {$w_{C^2_1}$};
\node [draw,circle] (g) at (3,-2) {};
\node at (3,-2.5) {$w_{C^1_2}$};
\node [draw,circle] (h) at (4,-2) {};
\node at (4,-2.5) {$w_{C^2_2}$};

\draw (a) ..controls (1.5,1).. (c);
\draw (a) ..controls (2,1.5).. (d);
\draw (b) -- (c);
\draw (b) ..controls (2.5,1).. (d);

\draw (b) -- (e);
\draw (b) -- (f);
\draw (c) -- (e);
\draw (c) -- (f);
\draw (a) -- (g);
\draw (a) -- (h);
\draw (c) -- (g);
\draw (g) -- (d);
\draw (c) -- (h);
\draw (h) -- (d);
\draw (h) -- (c);

     \end{tikzpicture}
  \caption{An illustration of the construction of the graph $G$ in Theorem~\ref{lem:no_kernel} for the  formula $\phi=(x_1 \vee \overline{x}_2) \wedge (\overline{x}_1)$. Here $C_1=(x_1 \vee \overline{x}_2)$ and $C_2=(\overline{x}_1)$.}
 \label{fig:no_kernel}
  \end{figure}
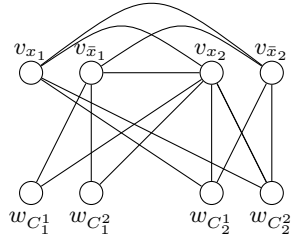
   
Let $V_X=\{v_{x},v_{\bar{x}}~|~x\in X\}$ and 
$V_{\cal C}=\{w_C^1,w_C^2~|~C\in{\cal C}\}$. 
Note that the set of vertices $V_{\cal C}$ corresponding to the clauses forms an independent set in $G$.
First we show that $G$ is a perfect graph. 

\begin{claim1}
\label{claim:no_antihole}
The graph $\overline{G}$ does not contain an induced odd cycle of length $\geq 5$. 
\end{claim1}
\begin{proof}
We first prove that there is no path of length $2$ (path on $3$ vertices) 
in $\overline{G[}V_X]$. Note that $E(\overline{G})=\{(v_x,v_{\bar{x}})~|~ x\in X\}$. This implies that the degree of each vertex in 
the graph $\overline{G}[V_X]$ is exactly equal to $1$. Hence. there is no path 
of length $2$ in $\overline{G}[V_X]$. Also, since $V_{\cal C}$ forms a clique in $\overline{G}$, any induced cycle 
of length at least $5$ in  $\overline{G}$ will either contain  a vertex or an edge from $V_{\cal C}$. 

Let $C'$ be an induced odd cycle of length at least $5$ in $\overline{G}$. There will be at most two vertices from $V_{\cal C}$ 
which are part of $C'$ and these vertices will appear consecutively 
in $C'$. This implies that $C'$ contains a path of length at least $2$ 
using only vertices from $V_X$ in $\overline{G}$, which is a contradiction. 
\qed
\end{proof}
\begin{claim1}
\label{claim:no_hole}
The graph $G$ does not contain an induced odd cycle of length $\geq 5$. 
\end{claim1}

\begin{proof}
We first show that any induced odd cycle of length at least $5$ 
can contain at most $3$ vertices  from $V_X$. Suppose not. 
Let $C'$ be an  induced odd cycle of length at least $5$ 
such that $\vert V(C')\cap V_X\vert \geq 4$. 
Let $v_w,v_x,v_y,v_z$ be four distinct vertices 
from $V(C')\cap V_X$ appearing in that order, if we go around the cycle in clockwise manner. 
That is,  there are paths 
$v_w-v_x$, $v_x-v_y$, $v_y-v_z$  in $C'$. Since 
$C'$ is an induced cycle, there is no edge $(v_w,v_y)$ in $E(G)$. 
This implies that $y=\bar{w}$. By similar arguments we can show that 
$x=\bar{z}$. This implies that $v_wv_xv_yv_zv_w$ form a cycle of length 
$4$ in $G$, contradicting the fact that $C'$ is induced odd cycle 
containing $v_w,v_x,v_y$ and $v_z$.  Hence 
any induced odd cycle of length at least $5$ 
can contain at most $3$ vertices  from $V_X$. 

Since $V_{\cal C}$ is an independent set in $G$, no two vertices 
from $V_{\cal C}$ can occur as consecutive vertices in any cycle. 
Let $C'$ be an induced odd cycle of length at least $5$ in $G$. 
Since $\vert V(C')\cap V_X\vert \leq 3$ and no two vertices 
from $V_{\cal C}$ appear as consecutive  vertices in $C'$, 
we have that $\vert V(C')\cap V_{\cal C}\vert \leq 2$. 
This implies that  the length of $C'$ is exactly equal to $5$ and $C'$ is of the form 
$v_xw_{C_1}^i v_y w_{C_2}^j v_z v_x $, where $i,j\in \{1,2\}$. 
Since $C'$ is an induced cycle $(v_x,v_y), (v_y,v_z) \notin E(G)$. 
This implies that $y=\bar{x}=z$ and hence $v_y=v_z$. This 
contradicts the fact that $C'$ is a cycle. This completes the proof 
of the claim.  
\qed
\end{proof}

Proposition~\ref{prop:perfect_graph_char} and Claims~\ref{claim:no_antihole} and \ref{claim:no_hole} imply that $G$ is a perfect graph. 
We now show that $(\phi,n)$ is a \YES{} instance of {\sc CNF-SAT} if and only if $(G,n,1)$ is a \YES{} instance of 
{\sc Partization Recognition}. 
 
First, suppose that $(\phi,n)$ is a \YES{} instance of {\sc CNF-SAT}. Then there is an assignment $\tau$, such that every clause has at least one literal set to $1$. Let $f : {\cal C}\rightarrow X$ be a 
map that arbitrarily maps one such satisfying literal to each clause. 
Note that for a clause $C$, $(w_C^1,v_{f(C)}),(w_C^2,v_{f(C)})\notin E(G)$, 
because $f(C)\in C$.  
Now we construct $n$ independent sets as follows. 
For each literal $y$, if $\tau(y)=1$, let 
$I_y = \{w_C^1,w_C^2 ~\vert~ f(C) = y\} \cup \{v_y\} $.  
Since $V_{\cal C}$ is an independent set 
$I_y\setminus \{v_y\}$ is an independent set. 
Note that for all $w_C^i\in I_y$, $i\in\{1,2\}$ we have that 
$(w_C^i,v_y)\notin E(G)$, because $f(C)=y$ and $y\in C$. 
This implies that $I_y$ is an independent set. 
Since $\tau$ is an assignment, exactly one of the literals of each variable is assigned $1$ by $\tau$. Thus, in this way we form $n$ independent sets. Since $\tau$ is  a satisfying assignment for $\phi$, the function 
$f$ maps each clause $C$ to a literal which is assigned $1$ by $\tau$. 
This implies that all vertices in $V_{\cal C}$ are covered by the independent sets constructed above.   
The vertices in the graph $G$, which are not covered by the independent 
sets constructed,  correspond to the literals that have been set to $0$ by $\tau$. By construction of $G$ and by the definition of an assignment $\tau$, these vertices form a clique. Hence, $(G,n,1)$ is an  \YES{} instance of 
{\sc Partization Recognition}.

Conversely, suppose $(G,n,1)$ is a \YES{} instance of 
{\sc Partization Recognition}. 
Then there is an \rlpartition{} $\cal P$ of $G$. 
Let $I_1,\ldots,I_n$ be $n$ independent sets and $K$ be a clique 
in the \rlpartition{} $\cal P$. 
It is not possible, by construction of $G$, that there is a variable $x$ such that both $v_x$ and $v_{\bar{x}}$ belong to the clique $K$, 
because $(v_x,v_{\bar{x}})\notin E(G)$.
 As there is only one clique in $\cal P$, at most one literal of each variable can be contained in the clique $K$ of $\cal P$. Hence, for each variable 
$x$   either $v_x$ or $v_{\bar x}$ is part of an independent set in $\cal P$. 
Furthermore, since for two literals $p$ and $q$ such that $p\neq \bar{q}$, 
$(v_p,v_q)\in E(G)$, any independent set $I$ in $\cal P$ can not contain 
both $v_p$ and $v_q$. This implies that each of independent set among 
$n$ independent set can be identified by a variable $x\in X$. Since there 
are only $n$ independent sets in $\cal P$, there cannot be a variable $x$ such that both $v_x$ and $v_{\bar{x}}$ are part of distinct independent sets in $\cal P$.  Thus the construction of $G$ forces only two possibilities for each variable: 
\begin{itemize}
\item[(a)] there is exactly one literal vertex that is part of an independent set  while the other one belong to the clique $K$, or
\item[(b)] both literals together form an independent set. 
\end{itemize}

Now we construct an assignment $\tau$ and show that  
$\tau$ is a satisfying assignment for $\phi$. For a literal $z$, 
if $v_z\in K$, then we set $\tau(z)=0$. If for a variable $x$, 
both vertices $v_x$ and $v_{\bar{x}}$ do not belong to $K$ 
then we set  $\tau(x)=1$. 
Now we show that $\tau$ is a satisfying assignment for $\phi$.  Let $C$ be a clause in the formula $\phi$.  Since $(w_C^1,w_C^2)\notin E(G)$ at least one of $w_C^1$ or $w_C^2$ 
belongs to an independent set $I$ in $\cal P$. Let 
$w_C^i\in I$, where $i\in\{1,2\}$. We have shown that each independent set contains at least one vertex corresponding to a  
literal $y$. Since $v_y$ and $w_C^i$ belong to $I$, we have that $(v_z,w_C^i)\notin E(G)$. This implies that 
$y \in C$. Furthermore, this implies that $v_{\bar{y}}\notin I$ as no clause contains both $y$ and $\bar{y}$. Hence, we have that 
$v_{\bar{y}}$ is in $K$ and thus $\tau(\bar{y})=0$. This implies that $y$ is set to $1$ and hence the clause $C$ is satisfied.  This proves that $\tau$ is a satisfying assignment for $\phi$. 
\qed
\end{proof}

Note that {\sc Partization Recognition} is same as \rlpartization{}, when $k=0$. Hence we get the following corollary.    

\begin{corollary}
 \rlpartization{} parameterized by $k+r+\ell$ on perfect graphs does not have a polynomial kernel unless \containment. 
\end{corollary}

However for {\sc Vertex $(r,\ell)$ Partization} we give a polynomial kernel as mentioned 
in the following theorem and its proof can be found in Appendix~\ref{sec:polykernel}.

\begin{theorem}
\label{thm:poly_kernel}
{\sc Vertex $(r,\ell)$ Partization} on perfect graphs 
admits a kernel of size $k^{\Oh(d)}$ 
and  has  an algorithm with running time $d^k n^{\Oh(d)}$. Here, $d=f(r,\ell)$. 
\end{theorem}


\section{Polynomial kernel when $r$ and $\ell$ are constants}
\label{sec:polykernel}

We saw that there is no polynomial kernel for \rlpartization{}, unless \containment. 
The parameter for this problem is $k+r+\ell$, where the size of the deletion set is at most $k$ and the final graph 
is an \rlgraph. In this section, we consider the {\sc Vertex $(r,\ell)$ Partization} problem on perfect graphs, which is 
a special case of \rlpartization{} on perfect graphs. Here, for a given pair of fixed positive constants $(r,\ell)$, we take 
a perfect graph $G$ and a positive integer $k$ as input and decide whether there is a vertex set $S$ of size at most 
$k$ the deletion of which results in an \rlgraph. This simplified problem has a polynomial kernel, as shown below.

We first observe that when perfect graphs are \rlgraph{s}, this class coincides with another graph class, namely the 
class of perfect graphs that are \rlsplitgraph{s}. The notion of \rlsplitgraph{s} was introduced in~\cite{Gyarfas98}. 
\begin{definition}[\rlsplitgraph]
 A graph $G$ is an \rlsplitgraph{} if its vertex set can be partitioned into $V_1$ and $V_2$ such that the size of a largest clique in $G[V_1]$ is bounded by $r$
 and the size of a largest independent set in $G[V_2]$ is bounded by $\ell$. We call such a 
 partition $(V_1,V_2)$ as an \rlsplitpartition.
\end{definition}
From the definition of \rlsplitgraph, it follows that any \rlgraph{}  is also 
an  \rlsplitgraph. 
\begin{lemma}
 Let $G$ be a perfect graph. If $G$ is an \rlsplitgraph, then $G$ is also an 
 \rlgraph.
\end{lemma}

\begin{proof}
Since $G$ is a perfect graph, for any induced subgraph $G'$ of  $G$, 
the chromatic number of $G'$ ($\chi(G')$) is equal to the 
cardinality of a maximum sized clique of $G'$ ($\omega(G')$). 
We know that $G$ is an \rlsplitgraph. 
Let $(P_1,P_2)$ be an \rlsplitpartition{}  with $\omega(G[P_1]) \leq r$ and  
$\alpha(G[P_2]) \leq \ell$
Now we show that indeed $(P_1,P_2)$ is an \rlpartition{} of $G$. 
Since $G$ is a perfect graph, the graphs $G[P_1]$ and $\overline{G}[P_2]$ are perfect 
graphs. Since $G[P_1]$ is a perfect graph and $w(G[P_1])\leq r$, we have that 
$\chi(G[P_1])\leq r$. This implies that $P_1$ can be partitioned into 
$r$ independent sets. Since $\overline{G}[P_2]$ is a perfect graph and 
$\alpha(G[P_2])\leq \ell$, we have that $w(\overline{G}[P_2])\leq \ell$ and 
so $\chi(\overline{G}[P_2])\leq \ell$. This implies that $P_2$ 
can be partitioned into $\ell$ sets such that each set is independent in 
$\overline{G}[P_2]$. Hence $P_2$ can be partitioned in to $\ell$ 
cliques in $G[P_2]$. So $(P_1,P_2)$ is an \rlpartition{} of $G$. 
This completes the proof of the lemma. 
%
\qed
\end{proof}

For any fixed $r$ and $\ell$, there is a finite forbidden set $\F_{r,\ell}$ for \rlsplitgraph{s}~\cite{Gyarfas98}. 
That is, a graph $G$ is an \rlsplitgraph{} if and only if $G$ does not contain any graph $H\in \F_{r,\ell}$ as an 
induced subgraph. The size of the largest forbidden graph is bounded by $f(r,\ell)$, $f$ being a function given 
in~\cite{Gyarfas98}. Since $f(r,\ell)$ is a constant, it is possible to compute the forbidden set $\F_{r,\ell}$ 
in polynomial time. 
Thus, the class of perfect \rlgraph{s} also has a finite forbidden characterization. This implies that 
{\sc Vertex $(r,\ell)$ Partization} on perfect graphs reduces to the $d$-{\sc Hitting Set} problem, 
where $d$ is the constant $f(r,\ell)$. In an equivalent $d$-{\sc Hitting Set} instance, the universe 
will be the set of vertices of the input graph $G$, while the family of sets will be the vertex sets of induced 
subgraphs of $G$ that are isomorphic to a forbidden graph. The set sizes are bounded by 
$f(r,\ell)$. 
Hence, by \cite{Abu-Khzam10}, this problem has a polynomial kernel. This gives us the following theorem.

\medskip
\noindent
{\bf Theorem~\ref{thm:poly_kernel}.}
{\em {\sc Vertex $(r,\ell)$ Partization} on perfect graphs 
admits a kernel of size $k^{\Oh(d)}$ 
and  has  an algorithm with running time $d^k n^{\Oh(d)}$. Here, $d=f(r,\ell)$. 
}

  \section{Conclusion}
  In this paper we studied the {\sc Vertex Partization} problem on perfect graphs and showed that it is \FPT\ and does not admit a polynomial kernel. Furthermore, we also observed that  {\sc Vertex $(r,\ell)$ Partization} has an induced  finite forbidden characterization and utilized that to give a faster \FPT\ algorithm and a polynomial kernel for the problem.  However,  the algorithms for  
  {\sc Vertex $(r,\ell)$ Partization} has a factor of $n^{\Oh(d)}$, where $d$ depends on the size of a largest graph in the finite forbidden set. It would  be interesting to replace  the factor $n^{\Oh(d)}$ by $\tau(d) \cdot n^{\Oh(1)}$. 


\bibliography{references,ref}

%



\end{document}